\documentclass{article}%
\usepackage{graphicx}
\usepackage{amsmath}
\usepackage{amsfonts}
\usepackage{amssymb}%
\providecommand{\U}[1]{\protect\rule{.1in}{.1in}}
\newtheorem{theorem}{Theorem}

\newtheorem{definition}[theorem]{Definition}
\newtheorem{example}[theorem]{Example}

\newtheorem{lemma}[theorem]{Lemma}

\newtheorem{proposition}[theorem]{Proposition}
\newtheorem{remark}[theorem]{Remark}

\newenvironment{proof}[1][Proof]{\textbf{#1.} }{\ \rule{0.5em}{0.5em}}
\begin{document}

\title{\textbf{On the geometrical interpretation of locality in anomaly cancellation}}
\date{}
\author{\textsc{Roberto Ferreiro P\'{e}rez}\\Departamento de Econom\'{\i}a Financiera y Contabilidad I\\Facultad de Ciencias Econ\'omicas y Empresariales, UCM\\Campus de Somosaguas, 28223-Pozuelo de Alarc\'on, Spain\\\emph{E-mail:} \texttt{roferreiro@ccee.ucm.es}}
\maketitle

\begin{abstract}
A notion of local section of the determinant line bundle is defined giving
necessary and suficient conditions for anomaly cancellation compatible with
locality. This definition gives an intrinsic geometrical interpretation of the
local counterterms allowed in the renormalization program of quantum field
theory. For global anomalies the conditions for anomaly cancellation are
expressed in terms of the equivariant holonomy of the Bismut-Freed connection.

\end{abstract}


\smallskip


\section{Introduction}

Anomalies in quantum field theory admit a geometrical interpretation in terms
of determinant (or Pfaffian) line bundles (e.g. see \cite{ASZ}, \cite{AS},
\cite{freed}). In order to have a well defined quantum field theory the
determinant line bundle should be trivial. However, it is well known that this
condition is not sufficient for anomaly cancellation due to the locality
problem. Hence, to cancel the anomaly the determinant line bundle should admit
a special kind of section (a local section) corresponding to the local
counterterms allowed in the renormalization program of quantum field theory.
In this paper we study the geometrical interpretation of these local sections
in terms of the Bismut-Freed connection and we obtain necesary and sufficient
conditions for perturbative and global anomaly cancellation.

Let us explain in more detail the locality problem. We consider the action of
a group $\mathcal{G}$ on a bundle $E\rightarrow M$ over a compact $n$-manifold
$M$. Let $\{D_{s}:s\in\Gamma(E)\}$ be a $\mathcal{G}$-equivariant family of
elliptic operators acting on fermionic fields $\psi\in\Gamma(V)$ and
parametrized by $\Gamma(E)$. Then the Lagrangian density $\lambda_{D}%
(\psi,s)=\bar{\psi}iD_{s}\psi$ is $\mathcal{G}$-invariant, and hence the
classical action $\mathcal{A}_{\mathcal{L}}(\psi,s)=\int_{M}\bar{\psi}%
iD_{s}\psi$, is a $\mathcal{G}$-invariant function on $\Gamma(V)\times
\Gamma(E)$. However, at the quantum level, the corresponding partition
function defined by a formal fermionic path integral by $\mathcal{Z}%
(s)=\int\mathcal{D}\psi\mathcal{D}\bar{\psi}\exp\left(  -\int_{M}\bar{\psi
}iD_{s}\psi\right)  $ could fail to be $\mathcal{G}$-invariant. $\mathcal{Z}%
(s)$ can be defined in terms of regularized determinants of elliptic operators
but not in a unique way. There is an ambiguity in the definition of
$\mathcal{Z}(s)$ modulo the addition of local counterterms (e.g. see
\cite{Perrot}). Due to this ambiguity $\mathcal{Z}(s)$ could fail to be
$\mathcal{G}$-invariant. It can be seen that\ the modulus of $\mathcal{Z}(s)$
is $\mathcal{G}$-invariant. Hence we have $\mathcal{Z}(\phi\cdot
s)=\mathcal{Z}(s)\cdot\exp(2\pi i\cdot\alpha_{\phi}(s))$ where $\alpha
\colon\mathcal{G}\times\Gamma(E)\rightarrow\mathbb{R}/\mathbb{Z}$ satisfies
the cocycle condition $\alpha_{\phi_{2}\phi_{1}}(s)=\alpha_{\phi_{1}%
}(s)+\alpha_{\phi_{2}}(\phi_{1}s)$. Different definitions of $\mathcal{Z}(s)$
determine different cocycles, but they are cohomologous, in the sense that
they satisfy the condition $\alpha_{\phi}^{\prime}=\alpha_{\phi}+\phi^{\ast
}\theta-\theta$ for some $\theta\in\Omega^{0}(\Gamma(E))$. If $\alpha$ is an
exact cocycle, i.e., if there exists $\Lambda\in\Omega^{0}(\Gamma(E))$
satisfying
\begin{equation}
\alpha_{\phi}(s)=\Lambda(\phi\cdot s)-\Lambda(s) \tag{*}\label{Lambda}%
\end{equation}
we can define $\mathcal{Z}^{\prime}=\mathcal{Z}\cdot\exp(-2\pi i\Lambda)$ and
we have $\mathcal{Z}^{\prime}(\phi\cdot s)=\mathcal{Z}^{\prime}(s)$. Hence the
anomaly can be represented by a cohomology class in $H^{1}(\mathcal{G}%
,\Omega^{0}(\Gamma(E),\mathbb{R}/\mathbb{Z}))\simeq H^{1}(\mathcal{G}%
,\Omega^{0}(\Gamma(E))/\mathbb{Z})$ (e.g. see \cite{Blau}, \cite{Ekstrand},
\cite{Falqui}). For perturbative anomalies the group cohomology can be
replaced by Lie algebra cohomology. For $X\in\mathrm{Lie}\mathcal{G}$ we
define $\mathfrak{a}(X)=\left.  \frac{\delta\alpha_{\phi_{t}}}{\delta
t}\right\vert _{t=0}$ with $\phi_{t}=\exp(tX)$. If $\mathcal{G}$ is connected
condition (\ref{Lambda}) is equivalent to $\mathfrak{a}(X)=L_{X}\Lambda$.
Hence the condition for perturbative anomaly cancellation is equivalent to
$[\mathfrak{a}]=0\ $on $H^{1}(\mathrm{Lie}\mathcal{G},\Omega^{0}(\Gamma(E)))$
($\mathfrak{a}$ is closed by the Wess-Zumino consistency condition).

However, from the physical point of view that is not the end of the story.
Physics require that $\mathcal{Z}^{\prime}$ should be the fermionic path
integral of a Lagrangian density, and hence $\Lambda(s)$ should be a
\emph{local} functional, i.e., it should be of the form $\Lambda(s)=\int
_{M}\lambda(s)$, where $\lambda(s)(x)$ is a function of $s(x)$ and the
derivatives of $s$ at $x$. If that is the case, we can modify the Lagrangian
density to the effective Lagrangian $\mathcal{L}^{\prime}(s)=\bar{\psi}%
iD_{s}\psi-\lambda(s)$ and the partition function of $\mathcal{L}^{\prime}$ is
$\mathcal{Z}^{\prime}$. We say that the topological anomaly cancels if
condition (\ref{Lambda})\ is satisfied for a functional $\Lambda\in\Omega
^{0}(\Gamma(E))$, and that the physical anomaly cancels if condition
(\ref{Lambda}) is satisfied for a local functional $\Lambda\in\Omega
_{\mathrm{loc}}^{0}(\Gamma(E))$. Obviously the second condition implies the
first, but the converse is not true. Furthermore, if condition (\ref{Lambda})
is satisfied only for the connected component of the identity $\mathcal{G}%
_{0}$ on $\mathcal{G}$ we say that the perturbative\ (or local) anomaly
cancels. If it is satisfied for all the elements of $\mathcal{G}$ we say that
the global anomaly cancels. Hence the perturbative physical anomaly is
represented by a cohomology class in the local BRST cohomology $H^{1}%
(\mathrm{Lie}\mathcal{G},\Omega_{\mathrm{loc}}^{0}(\Gamma(E)))$ defined in
\cite{bonora}, and the global physical anomaly by a class in $H^{1}%
(\mathcal{G},\Omega_{\mathrm{loc}}^{0}(\Gamma(E))/\mathbb{Z})$.

The condition (\ref{Lambda}) admits the following geometrical interpretation.
The cocycle $\alpha$ determines an action on the trivial bundle $\mathcal{L}%
=\Gamma(E)\times\mathbb{C}\rightarrow\Gamma(E)$ by setting $\phi_{\mathcal{U}%
}(s,u)=(\phi(s),u\cdot\exp(2\pi i\alpha_{\phi}(s)))$ for $s\in\Gamma(E)$ and
$u\in\mathbb{C}$. If the action of $\mathcal{G}$ on $\Gamma(E)$ is free we can
consider the quotient bundle $\underline{\mathcal{L}}=(\Gamma(E)\times
\mathbb{C})/\mathcal{G}\rightarrow\Gamma(E)/\mathcal{G}$, and $\mathcal{Z}$
determines a section of $\underline{\mathcal{L}}$. Furthermore, if $\Lambda$
satisfies equation (\ref{Lambda}), then $\exp(2\pi i\Lambda)$ determines a
section of unitary norm of $\underline{\mathcal{L}}$, i.e., a section of the
principal $U(1)$-bundle $\underline{\mathcal{U}}=(\Gamma(E)\times
U(1))/\mathcal{G}\rightarrow\Gamma(E)/\mathcal{G}$. Hence topological anomaly
cancellation is equivalent to the existence of a section of $\underline
{\mathcal{U}}\rightarrow\Gamma(E)/\mathcal{G}$, and hence to the triviality of
$\underline{\mathcal{L}}$.

In \cite{AS} the bundle $\underline{\mathcal{L}}$ is identified with the
determinant line bundle of the family of operators. If $\mathcal{G}$ is
connected (i.e. for perturbative anomalies) a necessary and sufficient
condition for topological anomaly cancellation is that $c_{1}(\underline
{\mathcal{L}})=0$ on $H^{2}(\Gamma(E)/\mathcal{G})$. The advantage of this
approach to anomalies is that the Atiyah-Singer Index Theorem gives an
explicit expression for $c_{1}(\underline{\mathcal{L}})$ in terms of
characteristic classes. Furthermore, it gives the curvature $\mathrm{curv}%
(\underline{\Xi})$ of the Bismut-Freed connection $\underline{\Xi}$\ on
$\underline{\mathcal{L}}$ (e.g. see \cite{BF1}). Note the difference with the
approach based on group and Lie algebra cohomology, where $\alpha$ and
$\mathfrak{a}$ are defined only modulo exact terms and given by complicated
expressions on secondary characteristic classes. This approach also gives a
geometrical interpretation of the anomaly as a cohomology class in
$\Gamma(E)/\mathcal{G}$, and allows the use of topological tools\ in the study
of anomaly cancellation. However, due to the locality problem, the
cancellation of topological anomalies only gives necessary conditions for
physical anomaly cancellation, but they are not sufficient, i.e., anomalies in
field theory can exist even if the corresponding topological anomaly is
trivial. In order to take into account locality, it is proposed in
\cite{singer} (see also \cite{ASZ}) the problem of defining a notion of
\textquotedblleft local cohomology\textquotedblright\ giving necessary and
sufficient condition for physical anomaly cancellation. This problem was
solved in \cite{anomalies} for perturbative anomalies by the introduction of
local equivariant cohomology. In place of working with the cohomology of the
quotient $H^{2}(\Gamma(E)/\mathcal{G})$ we can also\ consider the
$\mathcal{G}$-equivariant cohomology $H_{\mathcal{G}}^{2}(\Gamma(E))$. For
free actions we have $H^{2}(\Gamma(E)/\mathcal{G})\simeq H_{\mathcal{G}}%
^{2}(\Gamma(E))$, and we can consider $\mathcal{L}=\Gamma(E)\times
\mathbb{C}\rightarrow\Gamma(E)$ as a $\mathcal{G}$-equivariant line bundle and
$\mathcal{U}=\Gamma(E)\times U(1)\rightarrow\Gamma(E)$ as a $\mathcal{G}%
$-equivariant $U(1)$-bundle. Furthermore, the $\mathcal{G}$-equivariant
curvature $\mathrm{curv}_{\mathcal{G}}(\Xi)$ of the Bismut-Freed connection
$\Xi$ on $\mathcal{L}$ is given by the equivariant Atiyah-Singer Index Theorem
(see \cite{FreedEqui}). One of the advantages of equivariant cohomology is
that it is also well defined for non-free actions. But the most important
advantage of $H_{\mathcal{G}}^{2}(\Gamma(E))$ with respect to $H^{2}%
(\Gamma(E)/\mathcal{G})$ is that $\mathrm{curv}_{\mathcal{G}}(\Xi)$ is a
\emph{local form}, whereas $\mathrm{curv}(\underline{\Xi})$ is non-local. In
\cite{anomalies} the notions of local forms $\Omega_{\mathrm{loc}}^{\bullet
}(\Gamma(E))$ and local equivariant forms $\Omega_{\mathrm{loc},\mathcal{G}%
}^{\bullet}(\Gamma(E))$ are defined in terms of the jet bundle of $E$. For
Gauge and gravitational anomalies we have $\mathrm{curv}_{\mathcal{G}}(\Xi
)\in\Omega_{\mathrm{loc}\mathbf{,}\mathcal{G}}^{2}(\Gamma(E))$. Furthermore,
the cancellation of the class of $\mathrm{curv}_{\mathcal{G}}(\Xi)$ on
$H_{\mathrm{loc}\mathbf{,}\mathcal{G}}^{2}(\Gamma(E))$ is equivalent to the
cancellation of the perturbative physical anomaly. This approach provides new
techniques for the study of anomaly cancellation as the local cohomology
$H_{\mathcal{G},\mathrm{loc}}^{2}(\Gamma(E))$ is very different to the
cohomology $H^{2}(\Gamma(E)/\mathcal{G})$ of the quotient space. It is shown
in \cite{VB} and \cite{anomalies} that $H_{\mathcal{G},\mathrm{loc}}%
^{2}(\Gamma(E))$ is related to the equivariant cohomology of jet bundles and
Gelfand-Fuks cohomology of formal vector fields.

The objective of this paper is to give a geometrical interpretation of the
preceding results and to generalize the results of \cite{anomalies} to global
anomalies. Our starting point for the study of anomaly cancellation is the
unitary determinant bundle $\mathcal{U}\rightarrow\Gamma(E)$ corresponding to
a $\mathcal{G}$-equivariant family of elliptic operators \cite{BF1}. We
consider $\mathcal{U}\rightarrow\Gamma(E)$ as a $\mathcal{G}$-equivariant
$U(1)$-bundle and the Bismut-Freed connection $\Xi$ is $\mathcal{G}%
$-invariant. We assume that $\mathcal{U}\rightarrow\Gamma(E)$ is a
topologically trivial bundle and hence admits global sections. To any section
$S$\ of $\mathcal{U}$ we associate a group cocycle $\alpha^{S}$ and a Lie
algebra cocycle $\mathfrak{a}^{S}$. In this way the different expressions of
the cocyle $\alpha$ and the integrated anomaly $\mathfrak{a}$ obtained from
perturbation theory correspond to different sections of $\mathcal{U}%
\rightarrow\Gamma(E)$. Furthermore, $S$ determines a trivialization of
$\mathcal{U}\rightarrow\Gamma(E)$, and in this trivialization any other
section is determined by a function of the form $\exp(2\pi i\Lambda)$. The
condition (\ref{Lambda}) for topological anomaly cancellation is equivalent to
the existence of a $\mathcal{G}$-equivariant section of the unitary
determinant bundle. We obtain necessary and sufficient conditions for the
existence of a $\mathcal{G}$-equivariant section in terms of the Bismut-Freed
connection $\Xi$. Perturbative anomalies are related to the $\mathcal{G}%
$-equivariant curvature of the connection. To study global anomalies we
introduce the concept of $\mathcal{G}$-equivariant holonomy of a connection
$\Xi$. We obtain necessary and sufficient conditions for global anomaly
cancellation in terms of the $\mathcal{G}$-equivariant holonomy.

To deal with the locality problem, we introduce the notion of local section of
the unitary determinant bundle. We say that a section $S\colon\Gamma
(E)\rightarrow\mathcal{U}$ is $\Xi$-local if $\rho^{S}=\frac{i}{2\pi}S^{\ast
}\Xi\in\Omega_{\mathrm{loc}}^{1}(\Gamma(E))$. We prove that\ local sections
exist and that for local sections the cocycles $\alpha^{S}$ and $\mathfrak{a}%
^{S}$ are local. Furthermore, in the trivialization determined by $S$ any
other $\Xi$-local section is given by a function of the form $\exp(2\pi
i\Lambda)$ for a local functional $\Lambda\in\Omega_{\mathrm{loc}}^{0}%
(\Gamma(E))$. The condition (\ref{Lambda}) is shown to be satisfied for a
local functional if and only if there exists a $\Xi$-local $\mathcal{G}%
$-equivariant\ section of $\mathcal{U}\rightarrow\Gamma(E)$. In this way an
intrinsic characterization of the cancellation of the physical anomaly is
obtained. Finally we obtain necessary and sufficient conditions for physical
anomaly cancellation in terms of the $\mathcal{G}$-equivariant curvature and
holonomy of the Bismut-Freed connection $\Xi$.

In \cite{LocUni} we combine our geometric\ characterization of global anomaly
cancellation with the results of \cite{VB}\ and\ \cite{WP} to analyze
gravitational anomaly cancellation (see Section \ref{CR} for more details).

\section{Topological anomalies}

In this Section we study conditions for the existence of $\mathcal{G}%
$-equivariant sections of a $\mathcal{G}$-equivariant $U(1)$-bundle
$\mathcal{U}\rightarrow N$. As commented in the Introduction, when they are
applied to the determinant line bundle, they give necessary and sufficient
conditions for topological anomaly cancellation.

\subsection{Topological anomalies and group cohomology}

We say that a $U(1)$-bundle $\mathcal{U}\rightarrow N$ is topologically
trivial if there exists an isomorphism of principal $U(1)$-bundles
$\mathcal{U}\simeq N\times U(1)$. We recall it is equivalent to give a section
or a trivialization of $\mathcal{U}$. If $S\colon N\rightarrow\mathcal{U}$ is
a section, then $\Psi_{S}\colon N\times U(1)\rightarrow\mathcal{U}$, $\Psi
_{S}(x,u)=S(x)\cdot u$ is a trivialization of $\mathcal{U}$\ that we call the
trivialization associated to $S$. Conversely, if $\Psi\colon N\times
U(1)\rightarrow\mathcal{U}$ is a trivialization of $\mathcal{U}$, then
$\Psi\circ S_{1}\circ\Psi_{N}^{-1}$ is a section of $\mathcal{U}$, where
$S_{1}$ is the section $S_{1}(x)=(x,1)$ of the trivial bundle $N\times
U(1)\rightarrow N$ and $\Psi_{N}\colon N\rightarrow N$ is the projection of
$\Psi$.

Let $\mathcal{G}$ be a group acting (on the left on $N$). A $\mathcal{G}%
$-equivariant $U(1)$-bundle is a $U(1)$-bundle in which $\mathcal{G}$ acts by
$U(1)$-automorphisms. We say that $\mathcal{U}\rightarrow N$ is a trivial
$\mathcal{G}$-equivariant $U(1)$-bundle if there exists a $\mathcal{G}%
$-equivariant isomorphism $\mathcal{U}\simeq N\times U(1)$ (where
$\mathcal{G}$ acts trivially in $U(1)$). This condition is equivalent to the
existence of a $\mathcal{G}$-equivariant section $S\colon N\rightarrow
\mathcal{U}$, i.e. such that $\phi_{\mathcal{U}}\circ S=S\circ\phi_{N}$ for
any $\phi\in\mathcal{G}$.

A necessary condition for a bundle to be a trivial $\mathcal{G}$-equivariant
$U(1)$-bundle is that it should be topologically trivial. We want to study
conditions for a topologically trivial bundle to be a trivial $\mathcal{G}%
$-equivariant bundle. Hence from now on we made the following assumption

(a1) $\mathcal{U}\rightarrow N$ is a $\mathcal{G}$-equivariant $U(1)$-bundle
that is topologically trivial, $N$ is connected and $H^{1}(N)=0$.

If $S\colon N\rightarrow\mathcal{U}$ is a section, then $\phi_{\mathcal{U}%
}^{-1}\circ S\circ\phi_{N}$ is also a section of $\mathcal{U}$ and we have
$\phi_{\mathcal{U}}^{-1}\circ S\circ\phi_{N}=S\cdot\exp(-2\pi i\alpha_{\phi
}^{S})$ for a function $\alpha_{\phi}^{S}\colon N\rightarrow\mathbb{R}%
/\mathbb{Z}$. We note that $S$ is a $\mathcal{G}$-equivariant section if and
only if $\alpha_{\phi}^{S}=0$ for any $\phi\in\mathcal{G}$.\ The function
$\alpha^{S}\colon\mathcal{G}\times N\rightarrow\mathbb{R}/\mathbb{Z}$
satisfies the following properties

\begin{lemma}
\label{cocyclo}a) \textbf{(cocycle condition)} We have $\alpha_{\phi^{\prime
}\phi}^{S}(x)=\alpha_{\phi}^{S}(x)+\alpha_{\phi^{\prime}}^{S}(\phi x)$ for any
$\phi$, $\phi^{\prime}\in\mathcal{G}$.

b) If $S^{\prime}(x)=S(x)\cdot\exp(2\pi i\Lambda(x))$ is another section of
$\mathcal{U}\rightarrow N$ we have $\alpha_{\phi}^{S^{\prime}}=\alpha_{\phi
}^{S}+\Lambda-\phi_{N}^{\ast}\Lambda$.
\end{lemma}

\begin{proof}
a) We have $(\phi^{\prime}\phi)_{\mathcal{U}}^{-1}\circ S\circ(\phi^{\prime
}\phi)_{N}=\phi_{\mathcal{U}}^{-1}\circ\left[  (\phi^{\prime})_{\mathcal{U}%
}^{-1}\circ S\circ(\phi^{\prime})_{N}\right]  \circ\phi_{N}$

$=\phi_{\mathcal{U}}^{-1}\circ\left[  S\cdot\exp(2\pi i\alpha_{\phi^{\prime}%
}^{S})\right]  \circ\phi_{N}=\left[  \phi_{\mathcal{U}}^{-1}\circ S\circ
\phi_{N}\right]  \cdot\exp(2\pi i\alpha_{\phi^{\prime}}^{S}\circ\phi_{N})$

$=S\cdot\exp(2\pi i\alpha_{\phi}^{S})\exp(2\pi i\alpha_{\phi^{\prime}}%
^{S}\circ\phi_{N})=S\cdot\exp(2\pi i(\alpha_{\phi}^{S}+\alpha_{\phi^{\prime}%
}^{S}\circ\phi_{N}))$

b) We have $\phi_{\mathcal{U}}^{-1}\circ S^{\prime}\circ\phi_{N}%
=\phi_{\mathcal{U}}^{-1}\circ\lbrack S\cdot\exp(2\pi i\Lambda)]\circ\phi
_{N}=\phi_{\mathcal{U}}^{-1}\circ S\circ\phi_{N}\cdot\exp(2\pi i(\Lambda
\circ\phi_{N}))$

$=S\cdot\exp(-2\pi i\alpha_{\phi})\cdot\exp(2\pi i\phi_{N}^{\ast}%
\Lambda)=S\cdot\exp(-2\pi i(\alpha_{\phi}-\phi_{N}^{\ast}\Lambda))$

$=S^{\prime}\cdot\exp(-2\pi i(\alpha_{\phi}-\phi_{N}^{\ast}\Lambda+\Lambda))$.
\end{proof}

\begin{remark}
\label{RemarkCocycle}\emph{In the particular case in which }$\alpha_{\phi}%
$\emph{ is constant for any} $\phi\in\mathcal{G}$ \emph{the cocycle condition
is equivalent to} $\alpha_{\phi^{\prime}\phi}=\alpha_{\phi}+\alpha
_{\phi^{\prime}}$\emph{, i.e.} $\alpha\colon\mathcal{G}\rightarrow
\mathbb{R}/\mathbb{Z}$ \emph{is a group homomorphism.}
\end{remark}

In the trivialization determined by the section $S$ the action of $\phi
\in\mathcal{G}$ on $N\times U(1)$ is given by $\phi_{\mathcal{U}}%
(x,u)=(\phi_{N}(x),u\cdot\exp(2\pi i\alpha_{\phi}^{S}(x)))$. Conversely, we
have the following result

\begin{proposition}
\label{converse}If $\alpha\colon\mathcal{G}\times N\rightarrow\mathbb{R}%
/\mathbb{Z}$ satisfies the cocycle condition, then $\phi\cdot(x,u)=(\phi
_{N}(x),\exp(2\pi i\alpha_{\phi}(x))\cdot u)$ defines a group action of
$\mathcal{G}$ on $\mathcal{U}=N\times U(1)$ and $\mathcal{U}\rightarrow N$ is
a $\mathcal{G}$-equivariant $U(1)$-bundle. For the section $S(x)=(x,1)$\ we
have $\alpha^{S}=\alpha$.
\end{proposition}

\begin{proof}
It is a group action as we have%
\begin{align*}
\phi^{\prime}\cdot(\phi\cdot(x,u))  &  =\phi^{\prime}\cdot((\phi_{N}%
(x),\exp(2\pi i\alpha_{\phi}(x))\cdot u))\\
&  =(\phi_{N}^{\prime}(\phi_{N}(x)),\exp(2\pi i(\alpha_{\phi}(x)+\alpha
_{\phi^{\prime}}(\phi_{N}x))\cdot u)\\
&  =((\phi^{\prime}\cdot\phi)_{N}(x),\exp(2\pi i\alpha_{\phi^{\prime}\cdot
\phi}(x))\cdot u)=(\phi^{\prime}\cdot\phi)\cdot(x,u)).
\end{align*}

\end{proof}

If $\alpha\colon N\rightarrow\mathbb{R}/\mathbb{Z}$ is a function, we define
its differential $\delta\alpha\in\Omega^{1}(N)$ by $\delta\alpha=-\frac
{i}{2\pi}\exp(-2\pi i\alpha)d(\exp(2\pi i\alpha))$. For $\alpha\colon
\mathbb{R}\rightarrow\mathbb{R}/\mathbb{Z}$\ \ we define $\frac{\delta\alpha
}{\delta t}\in\Omega^{0}(\mathbb{R)}$ by $\frac{\delta\alpha}{\delta t}%
=-\frac{i}{2\pi}\exp(-2\pi i\alpha)\frac{d}{dt}(\exp(2\pi i\alpha(t)))$.
If$\ \alpha=A\operatorname{mod}\mathbb{Z}$ for a real function $A\in\Omega
^{0}(N)$ then $\delta\alpha=dA$. For an arbitrary $\alpha$ the form
$\delta\alpha$ is closed but not necessarily exact. But if $H^{1}(N)=0$ we
have the following result

\begin{lemma}
\label{modZ}If $N$ is connected and $H^{1}(N)=0$, then for any $\alpha\colon
N\rightarrow\mathbb{R}/\mathbb{Z}$ there exists $A\in\Omega^{0}(N)$ such that
$\alpha=A\operatorname{mod}\mathbb{Z}$. Any other function satisfying this
condition is of the form $A+n$ with $n\in\mathbb{Z}$. Hence we have
$\Omega^{0}(N,\mathbb{R}/\mathbb{Z})\simeq\Omega^{0}(N)/\mathbb{\mathbb{Z}}$
\end{lemma}

\begin{proof}
As $d(\delta\alpha)=0$ and $H^{1}(N)=0$, there exist $A^{\prime}\in\Omega
^{0}(N)$ such that $dA^{\prime}=\delta\alpha$. The function $\exp(2\pi
i\alpha)\exp(-2\pi iA^{\prime})$ is constant and has modulus one, and hence
there exists $r\in\mathbb{R}$ such that $\exp(2\pi i\alpha)=\exp(2\pi
i(A^{\prime}+r))$. We can take $A=A^{\prime}+r$. The constant $r$ is unique
modulo $\mathbb{Z}$, and the result follows.
\end{proof}

Hence the cocycle $\alpha$ can be considered as a map $\alpha\colon
\mathcal{G}\rightarrow\Omega^{0}(N)/\mathbb{\mathbb{Z}}$. We denote by
$Z^{1}(\mathcal{G},\Omega^{0}(N)/\mathbb{\mathbb{Z}}\mathbf{)}$ the space of
maps $\alpha\colon\mathcal{G}\rightarrow\Omega^{0}(N)/\mathbb{\mathbb{Z}}$
satisfying the cocycle condition $\alpha_{\phi^{\prime}\phi}(x)=\alpha_{\phi
}(x)+\alpha_{\phi^{\prime}}(\phi x)$ and by $B^{1}(\mathcal{G},\Omega
^{0}(N)/\mathbb{\mathbb{Z}})$ the exact cocycles of the form $\alpha_{\phi
}=\phi_{N}^{\ast}\theta-\theta$ for a function $\theta\in\Omega^{0}%
(N\mathbb{)}$. The group cohomology is defined by $H^{1}(\mathcal{G}%
,\Omega^{0}(N)/\mathbb{\mathbb{Z}})=Z^{1}(\mathcal{G},\Omega^{0}%
(N)/\mathbb{\mathbb{Z}})/B^{1}(\mathcal{G},\Omega^{0}(N)/\mathbb{\mathbb{Z}}%
)$. As $H^{1}(N)=0,$ from Lemmas \ref{cocyclo}\ and \ref{modZ} it
follows\ that $\alpha^{S}$ determines a cohomology class in $H^{1}%
(\mathcal{G},\Omega^{0}(N)/\mathbb{\mathbb{Z})}$ that does not depend on the
section $S$ chosen. We denote this class by $[\alpha^{\mathcal{U}}]\in
H^{1}(\mathcal{G},\Omega^{0}(N)/\mathbb{Z)}$ and we have the following

\begin{proposition}
\label{SectionH1}If $\mathcal{U}\rightarrow N$ is a topologically trivial
$\mathcal{G}$-equivariant $U(1)$-bundle, and $H^{1}(N)=0$ then the following
conditions are equivalent

a) $\mathcal{U}\rightarrow N$ is a trivial $\mathcal{G}$-equivariant $U(1)$-bundle.

b) There exists a $\mathcal{G}$-equivariant section $S\colon N\rightarrow
\mathcal{U}$.

c) $[\alpha^{\mathcal{U}}]=0$ on $H^{1}(\mathcal{G},\Omega^{0}(N)/\mathbb{Z)}$.
\end{proposition}

\begin{proof}
That a) and b) are equivalent follows from the equivalence between
trivializations and sections.

b)$\Rightarrow$c) As $S\colon N\rightarrow\mathcal{U}$ is $\mathcal{G}%
$-equivariant we have $\alpha^{S}=0$, and hence $[\alpha^{\mathcal{U}}]=0.$

c)$\Rightarrow$b) If $\alpha^{\mathcal{U}}=0$ on $H^{1}(\mathcal{G},\Omega
^{0}(N)/\mathbb{Z)}$ we chose a section $S\colon N\rightarrow\mathcal{U}$ and
we have $\alpha^{S}=\phi_{N}^{\ast}\theta-\theta$ for $\theta\in\Omega
^{0}(N\mathbb{)}$. We define the section $S^{\prime}=S\cdot\exp(2\pi i\theta)$
and by Proposition \ref{cocyclo} b)\ we have $\alpha_{\phi}^{S^{\prime}%
}=\alpha_{\phi}^{S}-\phi_{N}^{\ast}\theta+\theta=0$, and hence $S^{\prime}$ is
$\mathcal{G}$-equivariant.
\end{proof}

\begin{remark}
\emph{If} $H^{1}(N)\neq0$ \emph{then the result is also true, but replacing
the cohomology }$H^{1}(\mathcal{G},\Omega^{0}(N)/\mathbb{Z)}$ \emph{with}
$H^{1}(\mathcal{G},\Omega^{0}(N,\mathbb{R}/\mathbb{Z))}$\emph{. We prefer to
work with }$\Omega^{0}(N)/\mathbb{Z}$ \emph{as it can be easily generalized to
local cohomology.}
\end{remark}

\subsection{Local topological anomalies and Lie algebra cohomology}

Let $\mathcal{G}_{0}$ be the connected component with the identity on
$\mathcal{G}$. Invariance under $\mathcal{G}_{0}$ can be determined in terms
of the Lie algebra $\mathrm{Lie}\mathcal{G}$. The action of $\mathcal{G}$ on
$\mathcal{U}$\ induces a homomorphism $\mathrm{Lie}\mathcal{G}\rightarrow
\mathfrak{X}(\mathcal{U})$. If $S\colon N\rightarrow\mathcal{U}$ is a section,
for any $X\in\mathrm{Lie}\mathcal{G}$ the vector $X_{\mathcal{U}%
}(S(x))-S_{\ast}(X_{N}(x))$ is vertical and hence we have
\begin{equation}
X_{\mathcal{U}}(S(x))-S_{\ast}(X_{N}(x))=2\pi\mathfrak{a}^{S}(X)(x)\xi
_{\mathcal{U}}(S(x)) \label{definitiona}%
\end{equation}
for a function $\mathfrak{a}^{S}(X)\in\Omega^{0}(N)$. The term $\mathfrak{a}%
^{S}$ is the infinitesimal variation of $\alpha^{S}$. Precisely we have the following

\begin{proposition}
\label{AlfaA}If $X\in\mathrm{Lie}\mathcal{G}$, and $\phi_{t}=\exp(tX)$ then we have

$\mathfrak{a}^{S}(X)=\left.  \tfrac{\delta\alpha_{\phi_{t}}}{\delta
t}\right\vert _{t=0}$
\end{proposition}

\begin{proof}
It follows by taking the derivative with respect to $t$ at $t=0$ on the
equation $(\phi_{t})_{\mathcal{U}}^{-1}\circ S\circ(\phi_{t})_{N}=S\cdot
\exp(-2\pi i\alpha_{\phi_{t}}^{S})$.
\end{proof}

We conclude that the section $S$ is $\mathcal{G}_{0}$-equivariant if and only
if $\mathfrak{a}^{S}(X)=0$ for any $X\in\mathrm{Lie}\mathcal{G}$.

Let us recall the definition of Lie algebra cohomology. If $b\in
\mathrm{Hom}(\mathrm{Lie}\mathcal{G},\Omega^{0}(N))$ we define $\partial
b(X,Y)=X_{N}(b(Y))-Y_{N}(b(X))-b([X,Y])$ for $X,Y\in\mathrm{Lie}\mathcal{G}$.
The closed elements $Z^{1}(\mathrm{Lie}\mathcal{G},\Omega^{0}(N))\subset
\mathrm{Hom}(\mathrm{Lie}\mathcal{G},\Omega^{0}(N))$ are those satisfying
$\partial b=0$. The exact elements $B^{1}(\mathrm{Lie}\mathcal{G},\Omega
^{0}(N))$ are those of the form $b(X)=L_{X}\Lambda$ for $\Lambda\in\Omega
^{0}(N)$. We define the Lie algebra cohomology by $H^{1}(\mathrm{Lie}%
\mathcal{G},\Omega^{0}(N))=Z^{1}(\mathrm{Lie}\mathcal{G},\Omega^{0}%
(N))/B^{1}(\mathrm{Lie}\mathcal{G},\Omega^{0}(N))$.

\begin{proposition}
\label{delatA}We have $\partial\mathfrak{a}^{S}=0$.
\end{proposition}

\begin{proof}
By the definition of $\mathfrak{a}^{S}$, for any $f\in\Omega^{0}(\mathcal{U})$
we have
\begin{equation}
X_{\mathcal{U}}(f)\circ S-X_{N}(f\circ S)=2\pi\mathfrak{a}^{S}(X)(\xi
_{\mathcal{U}}(f)\circ S) \label{derivation}%
\end{equation}
By using equation (\ref{derivation}) for the function $X_{\mathcal{U}}(f)$ we
obtain%
\begin{align}
Y_{\mathcal{U}}(X_{\mathcal{U}}(f))\circ S-Y_{N}(X_{\mathcal{U}}(f)\circ S)
&  =2\pi\mathfrak{a}^{S}(Y)(\xi_{\mathcal{U}}(X_{\mathcal{U}}(f))\circ
S)\nonumber\\
&  =2\pi\mathfrak{a}^{S}(Y)(X_{\mathcal{U}}(\xi_{\mathcal{U}}(f))\circ S)
\label{derivation1}%
\end{align}
(in the last equation we use that $[X_{\mathcal{U}},\xi_{\mathcal{U}}]=0$ as
$\mathcal{G}$ acts on $\mathcal{U}$ by \linebreak$U(1)$-automorphisms).
Equation (\ref{derivation}) applied to the function $\xi_{\mathcal{U}}(f)$
gives%
\begin{equation}
(X_{\mathcal{U}}(\xi_{\mathcal{U}}(f))\circ S)-X_{N}(\xi_{\mathcal{U}}(f)\circ
S)=2\pi\mathfrak{a}^{S}(X)(\xi_{\mathcal{U}}(\xi_{\mathcal{U}}(f))\circ S)
\label{derivation2}%
\end{equation}
And finally, by applying $Y_{N}$ to equation (\ref{derivation}) we obtain
\begin{equation}
Y_{N}(X_{\mathcal{U}}(f)\circ S)-Y_{N}(X_{N}(f\circ S))=2\pi Y_{N}%
(\mathfrak{a}^{S}(X))(\xi_{\mathcal{U}}(f)\circ S)+2\pi\mathfrak{a}%
^{S}(X)Y_{N}(\xi_{\mathcal{U}}(f)\circ S) \label{derivation3}%
\end{equation}

By using equations (\ref{derivation}), (\ref{derivation1}), (\ref{derivation2}%
) and (\ref{derivation3}) we obtain%
\begin{align*}
\mathfrak{a}^{S}([X,Y])(\xi_{\mathcal{U}}(f)\circ S))  &  =\tfrac{1}{2\pi
}([X_{\mathcal{U}},Y_{\mathcal{U}}](f)\circ S-[X_{N},Y_{N}](f\circ
S))(\xi_{\mathcal{U}}(f)\circ S))=\\
&  (Y_{N}(\mathfrak{a}^{S}(X))-X_{N}(\mathfrak{a}^{S}(Y)))(\xi_{\mathcal{U}%
}(f)\circ S)
\end{align*}
As $f$ is arbitrary we conclude that $\mathfrak{a}^{S}([X,Y])=X_{N}%
(\mathfrak{a}^{S}(Y))-Y_{N}(\mathfrak{a}^{S}(X))$ and $\partial\mathfrak{a}%
^{S}=0$.
\end{proof}

As $H^{1}(N)=0$, if $S$ is a section, then by Lemma \ref{modZ} any other
section $S^{\prime}$ can be expressed as $S^{\prime}=S\cdot\exp(2\pi iA)$ for
a function $A\in\Omega^{0}(M)$. The following result follows from Proposition
\ref{AlfaA} and Lemma \ref{cocyclo} b)

\begin{lemma}
\label{variacionSeccion}If $S^{\prime}=S\cdot\exp(2\pi iA)$ with $A\in
\Omega^{0}(M)$, then $\mathfrak{a}^{S^{\prime}}(X)=\mathfrak{a}^{S}(X)-L_{X}A$.
\end{lemma}

We conclude that the cohomology class of $\mathfrak{a}^{S}$ on $H^{1}%
(\mathrm{Lie}\mathcal{G},\Omega^{0}(N))$ does not depend on the section $S$
chosen. We denote this class by $[\mathfrak{a}^{\mathcal{U}}]\in
H^{1}(\mathrm{Lie}\mathcal{G},\Omega^{0}(N))$.

\section{Connections and anomaly cancellation}

In this Section we study the conditions of topological anomaly cancellation in
terms of a $\mathcal{G}$-invariant connection $\Xi$ on $\mathcal{U}\rightarrow
N$. We show that local anomalies are related to the $\mathcal{G}$-equivariant
curvature, while global anomalies are related to the $\mathcal{G}$-equivariant
holonomy of $\Xi$.

\subsection{Equivariant cohomology in the Cartan model}

First, we recall the definition of equivariant cohomology in the Cartan model
(\emph{e.g. }see \cite{BGV}). Suppose that we have a left action of a
connected Lie group $\mathcal{G}$ on a manifold $N$. We denote by
$H^{k}(N)^{\mathcal{G}}$ the cohomology of the space of $\mathcal{G}%
$-invariant forms $\Omega^{k}(N)^{\mathcal{G}}$ on $N$. Let $\Omega
_{\mathcal{G}}^{\bullet}(N)=\left(  \mathbf{S}^{\bullet}(\mathrm{Lie\,}%
\mathcal{G}^{\ast})\otimes\Omega^{\bullet}(N)\right)  ^{\mathcal{G}%
}=\mathcal{P}^{\bullet}(\mathrm{Lie\,}\mathcal{G},\Omega^{\bullet
}(N))^{\mathcal{G}}$ be the space of $\mathcal{G}$-invariant polynomials on
$\mathrm{Lie\,}\mathcal{G}$ with values in $\Omega^{\bullet}(N)$ with the
graduation $\deg(\alpha)=2k+r$ if $\alpha\in\mathcal{P}^{k}(\mathrm{Lie\,}%
\mathcal{G},\Omega^{r}(N))$. Let $D\colon\Omega_{\mathcal{G}}^{q}%
(N)\rightarrow\Omega_{\mathcal{G}}^{q+1}(N)$ be the Cartan differential,
$(D\alpha)(X)=d(\alpha(X))-\iota_{X_{N}}\alpha(X)$, $X\in\mathrm{Lie\,}%
\mathcal{G}$. On $\Omega_{\mathcal{G}}^{\bullet}(N)$ we have $D^{2}=0$, and
the equivariant cohomology (in the Cartan model) of $N$ with respect of the
action of $\mathcal{G}$ is defined as the cohomology of this complex.

A $\mathcal{G}$-equivariant $1$-form $\alpha\in\Omega_{\mathcal{G}}%
^{1}(\mathcal{N})$ is just a $\mathcal{G}$-invariant $1$-form $\alpha\in
\Omega^{1}(\mathcal{N})^{\mathcal{G}}$. It is $D$-closed if and only if it is
$\mathcal{G}$-basic, i.e., if $d\alpha=0$ and $\iota_{X_{N}}\alpha=0$ for any
$X\in\mathrm{Lie\,}\mathcal{G}$.

Let $\varpi\in\Omega_{\mathcal{G}}^{2}(N)$ be a $\mathcal{G}$-equivariant
$2$-form. Then we have $\varpi=\omega+\mu$ where $\omega\in\Omega^{2}(N)$ is
$\mathcal{G}$-invariant and $\mu\in\mathrm{Hom}\left(  \mathrm{Lie\,}%
\mathcal{G},\Omega^{0}(N)\right)  ^{\mathcal{G}}$. We have $D\omega
=0\;$if\ and only if $d\omega=0$, and $\iota_{X_{N}}\omega=d(\mu(X))\;$for
every $X\in\mathrm{Lie\,}\mathcal{G}$.

If $\pi\colon\mathcal{U}\rightarrow N$ is a principal $U(1)$ bundle and
$\Xi\in\Omega^{1}(\mathcal{U},i\mathbb{R)}$ is a connection then the curvature
form$\ \mathrm{curv}(\Xi)\in\Omega^{2}(N)$ is defined by the property
$\pi^{\ast}(\mathrm{curv}(\Xi))=\frac{i}{2\pi}d\Xi$. The (real) first Chern
class of $\mathcal{U}$ is the cohomology class of $\mathrm{curv}(\Xi)$. A
connection $\Xi\in\Omega^{1}(\mathcal{U},i\mathbb{R)}$ on $\mathcal{U}$ is
$\mathcal{G}$-invariant if $\phi_{\mathcal{U}}^{\ast}\Xi=\Xi$ for any $\phi
\in\mathcal{G}$. At the Lie algebra level this implies that $L_{X_{\mathcal{U}%
}}\Xi=0$ for any $X\in\mathrm{Lie}\mathcal{G}$, and the converse is true for
connected groups. If $\Xi$ is a $\mathcal{G}$-invariant connection then
$\frac{i}{2\pi}D(\Xi)$ projects onto a closed $\mathcal{G}$-equivariant
$2$-form $\mathrm{curv}_{\mathcal{G}}(\Xi)\in\Omega_{\mathcal{G}}^{2}(N))$
called the $\mathcal{G}$-equivariant curvature of $\Xi$. If $X\in
\mathrm{Lie}\mathcal{G}$ then we have $\mathrm{curv}_{\mathcal{G}}%
(\Xi)(X)=\mathrm{curv}(\Xi)+\mu^{\Xi}(X)$, where $\mu^{\Xi}(X)=-\frac{i}{2\pi
}\Xi(X_{\mathcal{U}})$. We say that a connection $\Xi$ is $\mathcal{G}$-flat
if $\mathrm{curv}_{\mathcal{G}}(\Xi)=0$. If $\Xi^{\prime}$ is another
$\mathcal{G}$-invariant connection we have $\Xi^{\prime}=\Xi-2\pi i(\pi^{\ast
}\lambda)$ for a $\mathcal{G}$-invariant $\lambda\in\Omega^{1}(N)$. Then
$\mathrm{curv}_{\mathcal{G}}(\Xi^{\prime})=\mathrm{curv}_{\mathcal{G}}%
(\Xi)+D\lambda$ and hence the equivariant cohomology class $[\mathrm{curv}%
_{\mathcal{G}}(\Xi)]\in\Omega_{\mathcal{G}}^{2}(N)$ does not depend on the
$\mathcal{G}$-invariant connection chosen.

\subsection{Local topological anomalies and equivariant curvature}

Let $\pi\colon\mathcal{U}\rightarrow N$ be a $\mathcal{G}$%
-equivariant\ principal $U(1)$-bundle and $\Xi$ a $\mathcal{G}$-invariant
connection. The Maurer-Cartan form on $U(1)$ is denoted by $\vartheta
=z^{-1}dz$, and $\xi\in\mathfrak{X}(U(1))$ is the $U(1)$-invariant vector
field $\xi(z)=iz$\ such that $\vartheta(\xi)=i$. We denote by $\xi
_{\mathcal{U}}\in\mathfrak{X}(\mathcal{U})$ the vector field on $\mathcal{U}$
corresponding to $\xi$. Given a section $S\colon N\rightarrow\mathcal{U}$, we
define $\rho^{S}=\frac{i}{2\pi}S^{\ast}\Xi\in\Omega^{1}(N)$. On the
trivialization $\Psi_{S}\colon N\times U(1)\rightarrow\mathcal{U}$ determined
by $S$ we have
\begin{equation}
\Psi_{S}^{\ast}\Xi=\vartheta-2\pi i\rho^{S} \label{trivialization}%
\end{equation}

Conversely, if $\rho\in\Omega^{1}(N)$ and $S$ is a section of $\mathcal{U}%
\rightarrow N$ then the form $\Xi=(\Psi_{S}^{-1})^{\ast}(\vartheta-2\pi
i\rho)$ is a connection form on $\mathcal{U}\rightarrow N$ with $\rho^{S}%
=\rho$.

The following result follows from the definitions of $\mathfrak{a}^{S}$,
$\rho^{S}$, $\mu^{\Xi}$ and equation (\ref{trivialization})

\begin{lemma}
\label{curvatura}We have

a) $\mathrm{curv}(\Xi)=d\rho^{S}$.

b) $\mu^{\Xi}(X)=-\rho^{S}(X_{N})+\mathfrak{a}^{S}(X)$ for any $X\in
\mathrm{Lie}\mathcal{G}$.

c) If $S^{\prime}=S\exp(2\pi i\Lambda)$ for $\Lambda\in\Omega^{0}(N)$, then
$\rho^{S^{\prime}}=\rho^{S}-d\Lambda$.
\end{lemma}

Note that $\mathfrak{a}^{S}$ and $\rho^{S}$ satisfy the following set of
equations, that are similar to the Stora-Zumino descent equations
\begin{align*}
\mathrm{curv}(\Xi)  &  =d\rho^{S}\\
L_{X}\rho^{S}+d(\mathfrak{a}^{S}(X))  &  =0.
\end{align*}

\begin{proposition}
\label{triviality}Let $\Xi$ be a $\mathcal{G}$-invariant connection on
$\mathcal{U}\rightarrow N$. Then for any $\rho\in\Omega^{1}(N)$ such that
$d\rho=\mathrm{curv}(\Xi)$ there exists a section $S\colon N\rightarrow
\mathcal{U}$ such that $\rho=\frac{i}{2\pi}S^{\ast}(\Xi)$. Any other section
satisfying this condition is of the form $S\cdot\exp(2\pi ir)$ for
$r\in\mathbb{R}$.
\end{proposition}

\begin{proof}
As $\mathcal{U}\rightarrow N$ is trivial, there exist a section $S_{0}\colon
N\rightarrow\mathcal{U}$. As $d\rho^{S_{0}}=\mathrm{curv}(\Xi)=d\rho$ and
$H^{1}(N)=0$ we have $\rho=\rho^{S_{0}}-d\Lambda$ for some $\Lambda\in
\Omega^{0}(N)$. We define the section $S=S_{0}\cdot\exp(2\pi i\Lambda)$ and we
have $\frac{i}{2\pi}S^{\ast}(\Xi)=\rho^{S_{0}}-d\Lambda=\rho$. If $S^{\prime}$
is another section satisfying this condition we have $S^{\prime}=S\cdot
\exp(2\pi ir)$ and $dr=\rho^{S}-\rho^{S^{\prime}}=0$.
\end{proof}

\begin{proposition}
\label{ThLocalTopological}If $\mathcal{U}\rightarrow N$ admits $\mathcal{G}%
_{0}$-invariant connections then the following conditions are equivalent

p$_{1}$) There exists a $\mathcal{G}_{0}$-equivariant section of
$\mathcal{U}\rightarrow N$.

p$_{2}$) $[\mathfrak{a}^{\mathcal{U}}]=0$ on the cohomology $H^{1}%
(\mathrm{Lie}\mathcal{G}$, $\Omega^{0}(N))$.

p$_{3}$) The first $\mathcal{G}_{0}$-equivariant Chern class $c_{1,\mathcal{G}%
_{0}}(\mathcal{U})\in H_{\mathcal{G}_{0}}^{2}(N)$ vanishes.
\end{proposition}

\begin{proof}
$p_{1}$)$\Rightarrow p_{3}$) If $S$ is a $\mathcal{G}_{0}$-equivariant section
then $\mathfrak{a}^{S}=0$, $\rho^{S}$ is $\mathcal{G}_{0}$-invariant and by
Proposition \ref{curvatura} we have $D\rho^{S}=0$.

$p_{3}$)$\Rightarrow p_{2}$) If $D\beta=\mathrm{curv}_{\mathcal{G}}(\Xi)$ for
a $\mathcal{G}_{0}$-invariant $\beta\in\Omega^{1}(N)$ then $d\beta
=\mathrm{curv}(\Xi)$ and $\iota_{X_{N}}\beta=-\mu^{\Xi}(X)$. By Proposition
\ref{triviality}\ there exists a section $S$ such that $\rho^{S}=\beta$, and
by Proposition \ref{curvatura} we have $\mathfrak{a}^{S}(X)=\mu^{\Xi}%
(X)+\beta(X_{N})=0$.

$p_{2}$)$\Rightarrow p_{1}$) If for a section we have $\mathfrak{a}%
^{S}(X)=L_{X}\Lambda$ for $\Lambda\in\Omega^{0}(N)$, we define $S^{\prime
}=S\exp(2\pi i\Lambda)$ and we have $\mathfrak{a}^{S^{\prime}}=\mathfrak{a}%
^{S}(X)-L_{X}\Lambda=0$ and $S^{\prime}$ is $\mathcal{G}_{0}$-equivariant.
\end{proof}

As commented in the Introduction anomalies are usually studied in terms of the
topology of the quotient space, but we prefer to work with equivariant
cohomology because it can be extended to local cohomology. However, the
topology of the quotient can be used to obtain necessary conditions for
anomaly cancellation.

Suppose that $\mathcal{H}\subset\mathcal{G}$ is a a subgroup that acts freely
on $N$ and we have a well defined quotient bundle $\mathcal{U}/\mathcal{H}%
\rightarrow N/\mathcal{H}$. If $\mathcal{U}\rightarrow N$ admits a
$\mathcal{G}$-equivariant section $S$, then $S$ is also $\mathcal{H}$
invariant and the first Chern class $c_{1}(\mathcal{U}/\mathcal{H})\in
H^{2}(N/\mathcal{H})$ vanishes. Hence we have the following

\begin{proposition}
If $c_{1}(\mathcal{U}/\mathcal{H})\neq0$ for some $\mathcal{H}\subset
\mathcal{G}$, then $\mathcal{U}\rightarrow N$ is a\ non trivial $\mathcal{G}%
$-equivariant $U(1)$-bundle.
\end{proposition}

This condition is frequently used to show that an anomaly does not cancel as
$c_{1}(\mathcal{U}/\mathcal{H})$ can be computed by using topological
techniques. However, we recall that, due to the locality problem, in this way
we obtain necessary conditions for physical anomaly cancellation, but they are
not sufficient.

\subsection{Global topological anomalies and equivariant holonomy}

In order to obtain conditions for\ a $\mathcal{G}$-equivariant $U(1)$-bundle
to be trivial in terms of a connection we need to obtain a condition analogous
to condition $p_{3}$) and valid for non connected groups. One possibility
could be to consider the integer equivariant Chern class. However, we need a
condition that should be able to be generalized to local cohomology, and that
is not the case for the integer cohomology. We show that this problem can be
solved by introducing the equivariant holonomy of a invariant connection.
Although it seems to be a natural concept, we have been unable\ to find a
detailed study of it\ in the literature. In this paper we give only the basic
facts needed for our characterization of anomaly cancellation. We left\ a more
detailed study of the equivariant holonomy for a separate paper.

\subsection{Equivariant Holonomy}

Let $\Xi$ be a $\mathcal{G}$-invariant connection on a $\mathcal{G}%
$-equivariant $U(1)$-bundle $\mathcal{U}\rightarrow N$ and let $I$ denote the
interval $[0,1]$. If $\phi\in\mathcal{G}$, we define $\mathcal{C}^{\phi
}=\{\gamma\colon I\rightarrow N:\gamma(1)=\phi(\gamma(0))\}$, and
$\mathcal{C}_{x}^{\phi}=\{\gamma\colon I\rightarrow N:\gamma(0)=x$ and
$\gamma(1)=\phi(x)\}$. If $\gamma\in\mathcal{C}_{x}^{\phi}$ and $y\in
\mathcal{U}$, with $\pi(y)=x$, we denote by $\overline{\gamma}\colon
I\rightarrow\mathcal{U}$ the $\Xi$-horizontal lift of $\gamma$ with
$\overline{\gamma}(0)=y$. We have $\pi(\overline{\gamma}(1))=\pi
(\phi_{\mathcal{U}}(y))=\phi(x)$, and hence there exists $h\in\mathbb{R}%
/\mathbb{Z}$ such that $\overline{\gamma}(1)=(\phi_{\mathcal{U}}(y))\exp(2\pi
ih)$. It can be easily seen that $h$ does not depend on the $y$ chosen and we
denote it by $\mathrm{hol}_{\phi}^{\Xi}(\gamma)$ and we call it the $\phi
$-equivariant holonomy of $\Xi$ on $\gamma$. Note that for $\phi
=1_{\mathcal{G}}$ the $1_{\mathcal{G}}$-equivariant holonomy coincides with
the ordinary holonomy. By using equation (\ref{trivialization}) we obtain the following

\begin{lemma}
\label{horizontalLift}On the trivialization determined by a section $S$ the
horizontal lift of $\gamma\colon I\rightarrow N$ with $\overline{\gamma
}(0)=(\gamma(0),u)$ is given by $\overline{\gamma}(s)=(\gamma(s),u\cdot
\exp(2\pi i%
{\textstyle\int\nolimits_{0}^{s}}
\rho_{\gamma(t)}^{S}\dot{\gamma}(t)dt))$.
\end{lemma}

From Lemma \ref{horizontalLift} and the definition of $\mathrm{hol}_{\phi
}^{\Xi}(\gamma)$ we conclude the following

\begin{proposition}
\label{loghol}If $S\colon N\rightarrow\mathcal{U}$ is a section of
$\mathcal{U}$ , then for any $\gamma\in\mathcal{C}_{x}^{\phi}$
\[
\mathrm{hol}_{\phi}^{\Xi}(\gamma)=\int_{\gamma}\rho^{S}-\alpha_{\phi}^{S}(x).
\]

\end{proposition}

\begin{remark}
\emph{The preceding Proposition can be used to give an alternative definition
of the cocycle }$\alpha_{\phi}$\emph{, by showing that }$\int_{\gamma}\rho
^{S}-\mathrm{hol}_{\phi}^{\Xi}(\gamma)$\emph{ does not depend on }$\gamma\in
C_{x}^{\phi}$\emph{. This is the approach used in \cite{CSconnections} to
define the Chern-Simons line bundles.}
\end{remark}

If $\gamma,\gamma^{\prime}\colon I\rightarrow N$ are curves on $N$, we define
the inverse curve $\overleftarrow{\gamma}(t)=\gamma(1-t)$, and if
$\gamma(1)=\gamma^{\prime}(0)$ we define $\gamma\ast\gamma^{\prime}\colon
I\rightarrow\mathbb{R}$ by $\gamma\ast\gamma^{\prime}(t)=\gamma(2t)$ for
$t\in\lbrack0,1/2]$ and $\gamma\ast\gamma^{\prime}(t)=\gamma^{\prime}(2t-1)$
for $t\in\lbrack1/2,1]$. If $\phi$ is a diffeomorphisms of $N$ then we define
$(\phi\cdot\gamma)(t)=\phi(\gamma(t)).$

\begin{proposition}
\label{localityAlfa}Let $\Xi$ be a $\mathcal{G}$-invariant connection,
$\phi,\phi^{\prime}\in\mathcal{G}$, $x,y\in N$, $\gamma\in\mathcal{C}%
_{x}^{\phi}$ and let $\zeta$ be a curve joining $y$ and $x$. We have

a) $\phi^{\prime}\cdot\gamma\in\mathcal{C}_{\phi x}^{\phi}$ and $\mathrm{hol}%
_{\phi}^{\Xi}(\phi^{\prime}\gamma)=\mathrm{hol}_{\phi}^{\Xi}(\gamma)$.

b) $\gamma^{\prime}=\zeta\ast\gamma\ast(\phi\cdot\overleftarrow{\zeta}%
)\in\mathcal{C}_{y}^{\phi}$ and $\mathrm{hol}_{\phi}^{\Xi}(\gamma^{\prime
})=\mathrm{hol}_{\phi}^{\Xi}(\gamma)$.

c) $\alpha_{\phi}^{S}(x)=\alpha_{\phi}^{S}(y)+\int_{\zeta}(\phi^{\ast}\rho
^{S}-\rho^{S})$.

d) $\delta\alpha_{\phi}^{S}=\phi^{\ast}\rho^{S}-\rho^{S}$.
\end{proposition}

\begin{proof}
a) and b) easily follows from the definition of $\mathrm{hol}_{\phi}^{\Xi
}(\gamma)$ and the invariance of $\Xi$ and d) follows from c). We prove c). If
$\gamma\in\mathcal{C}_{x}^{\phi}$ and $\gamma^{\prime}=\zeta\ast\gamma
\ast(\phi\cdot\overleftarrow{\zeta})$ then by using b) we obtain
\begin{align*}
\alpha_{\phi}^{S}(x)  &  =%
{\textstyle\int\nolimits_{\gamma}}
\rho^{S}-\mathrm{hol}_{\phi}^{\Xi}(\gamma)=%
{\textstyle\int\nolimits_{\gamma^{\prime}}}
\rho^{S}-%
{\textstyle\int\nolimits_{\zeta}}
\rho^{S}-%
{\textstyle\int\nolimits_{\phi\cdot\overleftarrow{\zeta}}}
\rho^{S}-\mathrm{hol}_{\phi}^{\Xi}(\gamma^{\prime})\\
&  =\alpha_{\phi}^{S}(y)-%
{\textstyle\int\nolimits_{\zeta}}
\rho^{S}+%
{\textstyle\int\nolimits_{\zeta}}
\phi^{\ast}\rho^{S}=\alpha_{\phi}^{S}(y)+%
{\textstyle\int\nolimits_{\zeta}}
(\phi^{\ast}\rho^{S}-\rho^{S})\text{.}%
\end{align*}

\end{proof}

The following results determines necessary and sufficient conditions for
global topological anomaly cancellation

\begin{theorem}
\label{ThGlobalTopological}If $\Xi$ is a $\mathcal{G}$-invariant connection
then the following conditions are equivalent

g$_{1}$) There exists a $\mathcal{G}$-equivariant section of $\mathcal{U}%
\rightarrow N$.

g$_{2}$) $[\alpha^{\mathcal{U}}]=0$ on $H^{1}(\mathcal{G}$, $\Omega
^{0}(N)/\mathbb{Z})$.

g$_{3}$) There exists $\beta\in\Omega^{1}(N)^{\mathcal{G}}$ such that
$\mathrm{hol}_{\phi}^{\Xi}(\gamma)=\int_{\gamma}\beta$ for any $\phi
\in\mathcal{G}$, and $\gamma\in\mathcal{C}^{\phi}$. Furthermore in that case
we have $D\beta=\mathrm{curv}_{\mathcal{G}}(\Xi)$.
\end{theorem}

\begin{proof}
We have seen in Proposition \ref{SectionH1} that g$_{1}$) and g$_{2}$) are equivalent.

g$_{1}$)$\Rightarrow$g$_{3}$) If $S$ is a $\mathcal{G}$-equivariant section of
$\mathcal{U}\rightarrow N$ then by definition we have $\alpha^{S}=0$ and
$\beta=\rho^{S}$ is $\mathcal{G}$-invariant. By Proposition \ref{loghol} for
any $\phi\in\mathcal{G}$, and $\gamma\in\mathcal{C}_{x}^{\phi}$ we have
$\mathrm{hol}_{\phi}^{\Xi}(\gamma)=-\alpha_{\phi}^{S}(x)+\int_{\gamma}\rho
^{S}=\int_{\gamma}\beta$.

g$_{3}$)$\Rightarrow$g$_{1}$) First we prove that if there exists $\beta
\in\Omega_{\mathcal{G}}^{1}(N)$ such that $\mathrm{hol}_{\phi}^{\Xi}%
(\gamma)=\int_{\gamma}\beta$, then we have $d\beta=\mathrm{curv}(\Xi)$. We
define $\chi=\mathrm{curv}(\Xi)-d\beta\in\Omega^{2}(N)$ and we choose a
section $S$ of $\mathcal{U}\rightarrow N$. For any loop $\gamma\in
\mathcal{C}_{x}^{1_{\mathcal{G}}}$ with $\gamma=\partial D$ we have $%
{\textstyle\int\nolimits_{D}}
\chi=%
{\textstyle\int\nolimits_{D}}
d(\rho^{S}-\beta)=%
{\textstyle\int\nolimits_{\gamma}}
(\rho^{S}-\beta)=%
{\textstyle\int\nolimits_{\gamma}}
\rho^{S}-\mathrm{hol}_{1_{\mathcal{G}}}^{\Xi}(\gamma)=\alpha_{1_{\mathcal{G}}%
}^{S}(x)=0$, and hence $\chi=0$, i.e., $\mathrm{curv}(\Xi)=d\beta$.

As $\mathrm{curv}(\Xi)=d\beta$, by Proposition \ref{triviality}\ there exists
a section $S^{\prime}$ such that $\rho^{S^{\prime}}=\beta$, and for any
$\phi\in\mathcal{G}$ and $\gamma\in\mathcal{C}_{x}^{\phi}$ we have$\ \alpha
_{\phi}^{S^{\prime}}(x)=\int_{\gamma}\rho^{S^{\prime}}-\mathrm{hol}_{\phi
}(\gamma)=0$, and $S^{\prime}$ is $\mathcal{G}$-invariant. Finally by
Proposition \ref{curvatura} we have $\iota_{X_{N}}\beta=\iota_{X_{N}}%
\rho^{S^{\prime}}=-\mu^{\Xi}(X)$ and hence $D\beta=\mathrm{curv}_{\mathcal{G}%
}(\Xi)$.
\end{proof}

\subsection{Global anomalies and equivariant flat connections\label{SectFlat}}

In the study of anomaly cancellation, we start with local anomalies because
they are easier to analyze. If the local anomaly cancels, then we study the
corresponding global anomaly. By Theorem \ref{ThLocalTopological}, if the
local topological anomaly cancels then there exists $\beta_{0}\in\Omega
^{1}(N)^{\mathcal{G}_{0}}$ such that $d\beta_{0}=\mathrm{curv}(\Xi)$ and
$\iota_{X_{N}}\beta_{0}=-\mu^{\Xi}(X)$. The first problem to cancel the global
anomaly\ is that the form $\beta_{0}$ does not need to be $\mathcal{G}%
$-invariant. As $\mathrm{curv}(\Xi)$ is $\mathcal{G}$-invariant we have
$d(\phi^{\ast}\beta_{0}-\beta_{0})=0$ and as $H^{1}(N)=0$ there exist
$\sigma_{\phi}\in\Omega^{0}(N)$ such that $d\sigma_{\phi}=\phi^{\ast}\beta
_{0}-\beta_{0}$. The function $\sigma_{\phi}^{\beta_{0}}$ is determined modulo
a constant. If $\phi\in\mathcal{G}_{0}$ we have $\phi^{\ast}\beta_{0}%
-\beta_{0}=0$ and we can take $\sigma_{\phi}^{\beta_{0}}$ constant. Hence
$\beta_{0}$ determines a map $\sigma^{\beta_{0}}\colon\mathcal{G}%
/\mathcal{G}_{0}\rightarrow\Omega^{0}(N)/\mathbb{R}$ and it is easily to see
that it satisfies the cocycle condition and hence defines an element
$[\sigma^{\beta_{0}}]\in H^{1}(\mathcal{G}/\mathcal{G}_{0},\Omega
^{0}(N)/\mathbb{R})$

\begin{proposition}
Let $\beta_{0}\in\Omega^{1}(N)^{\mathcal{G}_{0}}$ be a form such that
$d\beta_{0}\in\Omega^{1}(N)^{\mathcal{G}}$. If $H^{1}(N)=0$ then there exists
$\beta\in\Omega^{1}(N)^{\mathcal{G}}$ such that $d\beta=d\beta_{0}$ and
$\iota_{X}\beta=\iota_{X}\beta_{0}$ if and only if $[\sigma^{\beta_{0}}]=0$ on
$H^{1}(\mathcal{G}/\mathcal{G}_{0},\Omega^{0}(N)/\mathbb{R})$.
\end{proposition}

\begin{proof}
If $\beta_{0}$ is $\mathcal{G}$-invariant then $\sigma^{\beta_{0}}=0$.
Conversely, if $[\sigma^{\beta_{0}}]=0$ then there exists $\rho\in\Omega
^{0}(N)/\mathbb{R}$ such that $\sigma_{\phi}^{\beta_{0}}=\phi^{\ast}\rho-\rho
$. As $\sigma_{\phi}^{\beta_{0}}=0$ for any $\phi\in\mathcal{G}_{0}$ we have
$L_{X}\rho=0$.

If we define $\beta=\beta_{0}-d\rho$ then we have $d\beta=d\beta_{0}$. Also we
have $\phi^{\ast}(\beta)=\phi^{\ast}(\beta_{0})-d\phi^{\ast}(\rho)=\phi^{\ast
}(\beta_{0})-d\sigma_{\phi}^{\beta_{0}}-d\rho=\phi^{\ast}(\beta_{0}%
)-(\phi^{\ast}\beta_{0}-\beta_{0})-d\rho=\beta$ and $\iota_{X}\beta=\iota
_{X}\beta_{0}-\iota_{X}d\rho=\iota_{X}\beta_{0}.$
\end{proof}

Hence $[\sigma^{\beta_{0}}]\in H^{1}(\mathcal{G}/\mathcal{G}_{0},\Omega
^{0}(N)/\mathbb{R})$ is an obstruction to find a $\mathcal{G}$-equivariant
section of $\mathcal{U}\rightarrow N$. If this obstruction cancels, there
exists $\beta\in\Omega^{1}(N)^{\mathcal{G}}$ satisfying $\mathrm{curv}%
_{\mathcal{G}}(\Xi)=D\beta$. Then we define a new connection $\Xi^{\prime}%
=\Xi+2\pi i(\pi^{\ast}\beta),$ which is $\mathcal{G}$-invariant and we have
$\mathrm{curv}_{\mathcal{G}}(\Xi^{\prime})=0$, i.e. $\Xi^{\prime}$ is
$\mathcal{G}$-flat. Hence it is enough to study the case of $\mathcal{G}$-flat
connections. For $\mathcal{G}$-flat connections we have the following

\begin{proposition}
If $\Xi$ is a $\mathcal{G}$-flat connection then $\mathrm{hol}_{\phi}^{\Xi
}(\gamma)$ does not depend on $\gamma\in\mathcal{C}^{\phi}$, and
$\mathrm{hol}_{\phi}^{\Xi}(\gamma)=0$ for $\phi\in\mathcal{G}_{0}$.
\end{proposition}

\begin{proof}
By Proposition \ref{triviality} we can choose a section $S$ such that
$\rho^{S}=0$. By Proposition \ref{loghol}\ we have $\mathrm{hol}_{\phi}^{\Xi
}(\gamma)=-\alpha_{\phi}^{S}(x)$ that does not depend on $\gamma\in
\mathcal{C}_{x}^{\phi}$. Furthermore, by Proposition \ref{localityAlfa}
c)\ the holonomy does not depend on $x$.

By Proposition \ref{curvatura} we have $\mathfrak{a}^{S}=0$. Hence $S$ is
$\mathcal{G}_{0}$-equivariant and $\mathrm{hol}_{\phi}^{\Xi}(\gamma
)=\alpha_{\phi}^{S}=0$ for $\phi\in\mathcal{G}_{0}$.
\end{proof}

We conclude that $\mathrm{hol}_{\phi}^{\Xi}$ determines an element
$\kappa^{\Xi}\in\mathrm{Hom}(\mathcal{G}/\mathcal{G}_{0},\mathbb{R}%
/\mathbb{Z)}$, by setting $\kappa_{\phi}^{\Xi}=\mathrm{hol}_{\phi}^{\Xi
}(\gamma)$ for any $\gamma\in\mathcal{C}^{\phi}$. The following result shows
that any element of $\mathrm{Hom}(\mathcal{G}/\mathcal{G}_{0},\mathbb{R}%
/\mathbb{Z)}$ can be represented as the holonomy of a flat connection:

\begin{proposition}
\label{ExampleFlat}For any $h\in\mathrm{Hom}(\mathcal{G}/\mathcal{G}%
_{0},\mathbb{R}/\mathbb{Z})$ there exists a $\mathcal{G}$-flat connection
$\Xi$ on $\mathcal{U}\rightarrow N$ such that $\kappa^{\Xi}=h$.
\end{proposition}

\begin{proof}
Using the trivialization associated to a section the problem can be reduced to
a trivial bundle. On $\mathcal{U}=N\times U(1)$ we define $\alpha_{\phi
}(x)=h(\phi)$ (it satisfies the cocycle condition by Remark
\ref{RemarkCocycle}). By Proposition \ref{converse} it defines a $\mathcal{G}%
$-equivariant $U(1)$-bundle and $\Xi=\vartheta$ is a $\mathcal{G}$-flat
connection with $\kappa^{\Xi}=h$.
\end{proof}

If $\kappa^{\Xi}=0$ then the bundle is a trivial $\mathcal{G}$-equivariant
bundle, but if $\kappa^{\Xi}\neq0$ we cannot assert that the bundle is
nontrivial. The reason is that\ a trivial $\mathcal{G}$-equivariant bundle can
admit flat connections with nontrivial holonomy. An easy example is the following:

\begin{example}
We consider the trivial bundle $\mathcal{U}=\mathbb{R}\times U(1)$ and
$\mathcal{G}=\mathbb{Z}$ acting trivially on $U(1)$. We define $\alpha
_{n}(x)=\frac{n}{2}\operatorname{mod}\mathbb{Z}$ for $n\in\mathbb{Z}\ $and the
connection $\Xi=\vartheta$. For the curve $\gamma_{1}(s)=s$ we have
$\mathrm{hol}_{\phi}^{\Xi}(\gamma_{1})=\alpha_{1}(x)=\frac{1}{2}%
\operatorname{mod}\mathbb{Z}\neq0$ but $\mathcal{U}$\ is a
trivial$\ \mathbb{Z}$-equivariant\ bundle as for any $\gamma\in\mathcal{C}%
^{n}$\ we have $\kappa_{n}^{\Xi}=\int_{\gamma}\beta$ with $\beta=\frac{1}%
{2}dt\in\Omega^{1}(\mathbb{R})^{\mathbb{Z}}$.
\end{example}

If $\Xi$ is $\mathcal{G}$-flat, then by Theorem \ref{ThGlobalTopological} the
bundle $\mathcal{U}\rightarrow N$ is trivial if and only if there exists
$\beta\in\Omega^{1}(N)^{\mathcal{G}}$ such that $D\beta=0$ and $\kappa_{\phi
}^{\Xi}=\int_{\gamma}\beta$ for any $\gamma\in\mathcal{C}^{\phi}$. We study
this condition\ in more detail

\begin{proposition}
\label{cancelacionGobal}If $\beta\in\Omega^{1}(N)^{\mathcal{G}}$ satisfies
$D\beta=0$\ then $k_{\phi}^{\beta}=\int_{\gamma}\beta$ does not depend on
$\gamma\in\mathcal{C}^{\phi}$ and $k_{\phi}^{\beta}=0$ if $\phi\in
\mathcal{G}_{0}$. Furthermore, if $\beta=d\rho$ for a form $\rho\in\Omega
^{0}(N)$ then $k_{\phi}^{\beta}=\phi^{\ast}\rho-\rho$ for any $\phi
\in\mathcal{G}$. If there exists $\rho$ satisfying $\beta=d\rho$ and
$\mathcal{G}$-invariant then $k_{\phi}^{\beta}=0$ for any $\phi\in\mathcal{G}$.
\end{proposition}

\begin{proof}
As $d\beta=0$ we have $\beta=d\rho$ for a function $\rho\in\Omega^{1}(N)$
(determined up to a constant). Moreover, we have $k_{\phi}^{\beta}%
(x)=\int_{\gamma}\beta=\int_{\gamma}d\rho=\rho(\phi x)-\rho(x)$ that does not
depend on $\gamma\in\mathcal{C}_{x}^{\phi}$. Furthermore, $k_{\phi}^{\beta}$
does not depend on $x$ because $dk_{\phi}^{\beta}=d(\phi^{\ast}\rho-\rho
)=\phi^{\ast}d\rho-d\rho=\phi^{\ast}\beta-\beta=0$. We have $L_{X}\rho
=\iota_{X}\beta=0$ and hence $\rho\in\Omega^{1}(N)^{\mathcal{G}_{0}}$ and
$k_{\phi}^{\beta}=0$ for $\phi\in\mathcal{G}_{0}$.\ If $\rho\in\Omega
^{1}(N)^{\mathcal{G}}$ satisfies $\beta=d\rho$ then $k_{\phi}^{\beta}%
=\phi^{\ast}\rho-\rho=0$ for any $\phi\in\mathcal{G}$.
\end{proof}

Hence we have a well defined map $k\colon H_{\mathcal{G}}^{1}(N)\rightarrow
\mathrm{Hom}(\mathcal{G}/\mathcal{G}_{0},\mathbb{R}/\mathbb{Z)}$. We define
$K^{\mathcal{G}}(N)=k(H_{\mathcal{G}}^{1}(N))\subset\mathrm{Hom}%
(\mathcal{G}/\mathcal{G}_{0},\mathbb{R}/\mathbb{Z)}$ and we have the following

\begin{proposition}
If $\Xi$ is a $\mathcal{G}$-flat connection\ then $\mathcal{U}\rightarrow
N$\ is a trivial $\mathcal{G}$-equivariant\ bundle if and only if $\kappa
^{\Xi}\in K^{\mathcal{G}}(N)$.
\end{proposition}

\section{Locality and physical anomalies}

In this section we introduce the concepts of local forms and local cohomology
needed for the study the locality problem in anomaly cancellation. The local
cohomology is defined when $N$ is a submanifold of the space of sections
$\Gamma(E)$ of a bundle $E\rightarrow M$. It generalizes the concept of a
local functional to higher order forms. Furthermore, it is easily extended to
local $\mathcal{G}$-equivariant cohomology. We introduce the notions of local
connection and local section, and we show that the existence of local
$\mathcal{G}$-equivariant section is equivalent to physical anomaly cancellation.

\subsection{Local forms and Jet bundles\label{localforms}}

Let $p\colon E\rightarrow M$ be a bundle over a compact, oriented $n$-manifold
$M$ without boundary. We denote by $J^{r}E$ its $r$-jet bundle, and by
$J^{\infty}E$ the infinite jet bundle (see e.g.\ \cite{saunders} for the
details on the geometry of $J^{\infty}E$). We recall that the points on
$J^{\infty}E$ can be identified with the Taylor series of sections of $E$. Let
$\Gamma(E)$ be the manifold of global sections of $E$, that we assume to be
not empty. If $s\in\Gamma(E)$ then we denote by $j_{x}^{r}s$ (resp.
$j_{x}^{\infty}s$) the $k$-jet (resp. the $\infty$-jet) of $s$ at $x$.

We consider $\Gamma(E)$ as an infinite dimensional Frechet manifold, that we
assume to be not empty. Let $\mathrm{j}^{\infty}\colon M\times\Gamma
(E)\rightarrow J^{\infty}E$, $\mathrm{j}^{\infty}(x,s)=j_{x}^{\infty}s$ be the
evaluation map. In \cite{equiconn} it is defined a map $\Im\colon\Omega
^{n+k}(J^{\infty}E)\longrightarrow\Omega^{k}(\Gamma(E))$, by $\Im\lbrack
\alpha]=\int_{M}\left(  \mathrm{j}^{\infty}\right)  ^{\ast}\alpha$ for
$\alpha\in\Omega^{n+k}(J^{\infty}E)$. If $\alpha\in\Omega^{k}(J^{\infty}E)$
with $k<n$, we set $\Im\lbrack\alpha]=0$. We define the space of local
$k$-forms on $\Gamma(E)$ by $\Omega_{\mathrm{loc}}^{k}(\Gamma(E))=\Im
(\Omega^{n+k}(J^{\infty}E))\subset\Omega^{k}(\Gamma(E))$. We have $\Im\lbrack
d\alpha]=d\Im\lbrack\alpha]$, and hence if $\theta\in\Omega_{\mathrm{loc}}%
^{k}(\Gamma(E))$ then $d\theta\in\Omega_{\mathrm{loc}}^{k+1}(\Gamma(E))$. The
local cohomology of $\Gamma(E)$, $H_{\mathrm{loc}}^{\bullet}(\Gamma(E))$, is
the cohomology of $(\Omega_{\mathrm{loc}}^{\bullet}(\Gamma(E)),d)$ and we have
$H_{\mathrm{loc}}^{k}(\Gamma(E))\simeq H_{\mathrm{loc}}^{n+k}(\Gamma(E))$ for
$k>0$ (see \cite{VB}). If $\theta\in\Omega_{\mathrm{loc}}^{\bullet}%
(\Gamma(E))$ is closed, we denote by $[\theta]$ its cohomology class in
$H^{k}(\Gamma(E))$ and by $\{\theta\}$ its cohomology class in
$H_{\mathrm{loc}}^{k}(\Gamma(E))$.

\begin{lemma}
\label{Lemma} The constant functions on $\Gamma(E)$ are local functions.
\end{lemma}

\begin{proof}
If $c\in\mathbb{R}$ and $\omega\in\Omega^{n}(M)$ is a volume form on $M$ with
$v=\int_{M}\omega>0$, then $c=\Im\lbrack\frac{c}{v}q_{\infty}^{\ast}\omega]$,
where $q_{\infty}\colon J^{\infty}E\rightarrow M$ is the projection.
\end{proof}

\bigskip

A diffeomorphism $\phi\in\mathrm{Diff}E$ is said to be projectable if there
exists $\underline{\phi}\in\mathrm{Diff}M$ satisfying $\phi\circ
p=p\circ\underline{\phi}$. We denote by $\mathrm{Proj}E$ the space of
projectable diffeomorphism of $E$, and by $\mathrm{Proj}^{+}E$ the subgroup of
elements such that $\underline{\phi}$ is orientation preserving. Let
$\mathcal{G}$ be a Lie group acting on $E$ by elements $\mathrm{Proj}^{+}E$.
The integration operator extends to a map on equivariant differential forms
(see \cite{equiconn}) $\Im\colon\Omega_{\mathcal{G}}^{n+k}(J^{\infty
}E)\rightarrow\Omega_{\mathcal{G}}^{k}(\Gamma(E))$, by setting $(\Im
(\alpha))(X)=\Im(\alpha(X))$ for every $\alpha\in\Omega_{\mathcal{G}}%
^{n+k}(J^{\infty}E)$, $X\in\mathrm{Lie}\mathcal{G}$. The map $\Im$ induces a
homomorphism in equivariant cohomology $\Im\colon H_{\mathcal{G}}%
^{n+k}(J^{\infty}E)\rightarrow H_{\mathcal{G}}^{k}(\Gamma(E))$. We define the
space of local $\mathcal{G}$-equivariant $q$-forms on $\Gamma(E)$ by
$\Omega_{\mathcal{G},\mathrm{loc}}^{q}(\Gamma(E))=\bigoplus_{2k+r=q}%
(\mathcal{P}^{k}(\mathrm{Lie}\mathcal{G},\Omega_{\mathrm{loc}}^{r}%
(\Gamma(E))))^{\mathcal{G}}\subset\Omega_{\mathcal{G}}^{q}(\Gamma(E))$, and
the local $\mathcal{G}$-equivariant cohomology of $\Gamma(E)$, $H_{\mathcal{G}%
,\mathrm{loc}}^{\bullet}(\Gamma(E))$, as the cohomology of $(\Omega
_{\mathcal{G},\mathrm{loc}}^{\bullet}(\Gamma(E)),D)$. The map $\Im$ induces a
cochain map $\Im\colon\Omega_{\mathcal{G}}^{n+k}(J^{\infty}E)\rightarrow
\Omega_{\mathcal{G},\mathrm{loc}}^{k}(\Gamma(E))$.

\subsection{Local connections and sections}

Let $\mathcal{U}\rightarrow\Gamma(E)$ be a $\mathcal{G}$-equivariant
$U(1)$-bundle. On $\Gamma(E)$ we have the space of local forms, but we do not
have an analogous notion on $\mathcal{U}$. As a connection is a form on
$\mathcal{U}$, we define the notion of local connection in terms of its
equivariant curvature

\begin{definition}
We say that a $\mathcal{G}$-invariant connection $\Xi$ on $\mathcal{U}%
\rightarrow\Gamma(E)$ is local if $\mathrm{curv}_{\mathcal{G}}(\Xi
)=\mathrm{curv}(\Xi)+\mu^{\Xi}\in\Omega_{\mathcal{G},\mathrm{loc}}^{2}%
(\Gamma(E))$.
\end{definition}

The motivation for this definition is the equivariant Atiyah-Singer index
theorem for families (see \cite{FreedEqui}) that express the equivariant
curvature of the Bismut-Freed connection in terms of characteristic forms.\ It
is shown in \cite{anomalies} that for gravitational and gauge anomalies, the
Bismut-Freed connection on the determinant line bundle is a local connection.
We make the following assumptions:

(A1) We assume that $\mathcal{U}\rightarrow\Gamma(E)$ is topologically
trivial, $\Gamma(E)$ is connected and that $H^{1}(\Gamma(E))=0$.

As $\mathcal{U}\rightarrow\Gamma(E)$ is topologically trivial, we know that
$\mathrm{curv}(\Xi)\in\Omega_{\mathrm{loc}}^{2}(\Gamma(E))$ is exact. But we
cannot assert that it is the differential of a local form. Hence we need to
make the following assumption

(A2) $\Xi$ is a $\mathcal{G}$-invariant local connection and $\{\mathrm{curv}%
(\Xi)\}=0$ on $H_{\mathrm{loc}}^{2}(\Gamma(E))$, i.e. there exists $\rho
\in\Omega_{\mathrm{loc}}^{1}(\Gamma(E))$ such that $d\rho=\mathrm{curv}(\Xi)$.

A sufficient condition for (A2) is $H_{\mathrm{loc}}^{2}(\Gamma(E))\simeq
H^{n+2}(E)=0$. We also make the following assumption

(A3) $H_{\mathrm{loc}}^{1}(\Gamma(E))\simeq H^{n+1}(E)=0$.

\begin{definition}
We say that a section $S\colon\Gamma(E)\rightarrow\mathcal{U}$ is $\Xi$-local
if $\rho^{S}=\frac{i}{2\pi}S^{\ast}(\Xi)\in\Omega_{\mathrm{loc}}^{1}%
(\Gamma(E))$.
\end{definition}

We show below that the existence of local $\mathcal{G}$-invariant sections
characterizes physical anomaly cancellation and gives an intrinsic
characterization of the condition \ref{Lambda} in the Introduction. By
Proposition \ref{triviality}\ local sections exists if assumption (A2) is satisfied.

\begin{proposition}
\label{localSection}If $S$ is a $\Xi$-local section, then any other section
$S^{\prime}$ is $\Xi$-local if and only if is of the form $S^{\prime}%
=S\cdot\exp(2\pi i\cdot\Lambda)$ for $\Lambda\in\Omega_{\mathrm{loc}}%
^{0}(\Gamma(E))$.
\end{proposition}

\begin{proof}
If $S^{\prime}=S\cdot\exp(2\pi i\cdot\Lambda)$ for $\Lambda\in\Omega
_{\mathrm{loc}}^{0}(\Gamma(E))$ then by Proposition \ref{curvatura}\ \ we have
$\rho^{S^{\prime}}=\frac{i}{2\pi}(S^{\prime})^{\ast}(\Xi)=\rho^{S}-d\Lambda
\in\Omega_{\mathrm{loc}}^{1}(\Gamma(E))$.

Conversely, if $S^{\prime}$ is another local section then $\rho^{S^{\prime}%
}\in\Omega_{\mathrm{loc}}^{1}(\Gamma(E))$ and $d\rho^{S^{\prime}}=d\rho
^{S}=\mathrm{curv}(\Xi)$. By assumption (A3) $H_{\mathrm{loc}}^{1}%
(\Gamma(E))=0$ and hence there exists $\Lambda^{\prime}\in\Omega
_{\mathrm{loc}}^{0}(\Gamma(E))$ such that $\rho^{S^{\prime}}=\rho^{S}%
-d\Lambda^{\prime}$. By Propositions \ref{curvatura} and \ref{triviality}%
\ there exists $r\in\mathbb{R}$ such that $S^{\prime}=S\cdot\exp(2\pi
i\cdot\Lambda^{\prime})\exp(2\pi i\cdot r)=S\cdot\exp(2\pi i\cdot
(\Lambda^{\prime}+r))$ and we can take $\Lambda=\Lambda^{\prime}+r\in
\Omega_{\mathrm{loc}}^{0}(\Gamma(E))$ (it is local by Lemma \ref{Lemma}).
\end{proof}

\subsection{Local physical anomalies}

If $\Xi$ is a local connection and $S$ is $\Xi$-local section, then we have
$\mathfrak{a}^{S}(X)=\rho^{S}(X_{N})+\mu^{\Xi}(X)$ and hence $\mathfrak{a}%
^{S}\in\Omega^{1}(\mathrm{Lie}\mathcal{G},\Omega_{\mathrm{loc}}^{0}(N))$.

By Proposition \ref{delatA}\ we have $\partial\mathfrak{a}^{S}=0$%
.\ Furthermore, by Assumption (A3) and Propositions \ref{localSection}\ and
\ref{variacionSeccion} the cohomology class of $\mathfrak{a}^{S}$ on the local
BRST cohomology $H^{1}(\mathrm{Lie}\mathcal{G},\Omega_{\mathrm{loc}}%
^{0}(\Gamma(E)))$ does not depend on the $\Xi$-local section chosen. We denote
this class by $\{\mathfrak{a}^{\mathcal{U}}\}\in H^{1}(\mathrm{Lie}%
\mathcal{G},\Omega_{\mathrm{loc}}^{0}(\Gamma(E)))$.

The following Theorem is an extension of the results\ \cite{anomalies}\ that
shows that our definition of local section determines physical anomaly
cancellation. This is important in order to generalize this result to global anomalies.

\begin{theorem}
\label{ThLocalPhysical}If $\Xi$ is a $\mathcal{G}$-invariant local connection,
then the following conditions are equivalent

P$_{1}$) There exists a $\Xi$-local $\mathcal{G}_{0}$-invariant section of
$\mathcal{U}\rightarrow\Gamma(E)$.

P$_{2}$) $\{\mathfrak{a}^{\mathcal{U}}\}=0$ on the local BRST cohomology
$H^{1}(\mathrm{Lie}\mathcal{G}$, $\Omega_{\mathrm{loc}}^{0}(N))$

P$_{3}$) $\{\mathrm{curv}_{\mathcal{G}_{0}}(\Xi)\}=0$ on $H_{\mathcal{G}%
_{0},\mathrm{loc}}^{2}(\Gamma(E))$.
\end{theorem}

\begin{proof}
P$_{1}$)$\Rightarrow$P$_{3}$) If $S$ is a $\Xi$-local $\mathcal{G}_{0}%
$-invariant section $\rho^{S}\in\Omega_{\mathrm{loc}}^{1}(\Gamma
(E))^{\mathcal{G}_{0}}$ and $d\rho^{S}=\mathrm{curv}(\Xi)$ and $\mu^{\Xi
}(X)=\mathfrak{a}^{S}-\rho(X_{N})=-\rho(X_{N})$, as we have $\mathfrak{a}%
^{S}=0$ because $S$ is $\mathcal{G}_{0}$-invariant.

P$_{3}$)$\Rightarrow$P$_{2}$) Assume that there exists $\rho\in\Omega
^{1}(\Gamma(E))^{\mathcal{G}_{0}}$ such that $d\rho=\mathrm{curv}(\Xi)$ and
$\iota_{X_{N}}\rho=-\mu^{\Xi}(X)$. By Proposition \ref{triviality} there
exists\ a $\Xi$-local section $S$ such that $\rho^{S}=\rho$ and we have
$\mathfrak{a}^{S}(X)=\rho(X_{N})+\mu^{\Xi}(X)=0$.

P$_{2}$)$\Rightarrow$P$_{1}$) If $S\ $is a $\Xi$-local section and
$\mathfrak{a}^{S}(X)=L_{X}\Lambda$ for $\Lambda\in\Omega_{\mathrm{loc}}%
^{0}(\Gamma(E))$

then $S^{\prime}=S\cdot\exp(2\pi i\Lambda)$ is also local and by Proposition
\ref{variacionSeccion}\ we have\ $\mathfrak{a}^{S^{\prime}}(X)=\mathfrak{a}%
^{S}(X)-L_{X}\Lambda=0$ and $S^{\prime}$ is $\mathcal{G}_{0}$-invariant.
\end{proof}

\subsection{Global physical anomalies}

In this section we generalize the results of the previous section to obtain
necessary and sufficient conditions for global anomaly cancellation.

\begin{proposition}
If $S$ is a $\Xi$-local section then $\alpha^{S}\in\Omega^{1}(\mathcal{G}$,
$\Omega_{\mathrm{loc}}^{0}(\Gamma(E)\mathbb{)}/\mathbb{Z})$.
\end{proposition}

\begin{proof}
Let $\phi\in\mathcal{G}$ and fix $s_{0}\in\Gamma(E)$. We have $d(\phi^{\ast
}\rho^{S}-\rho^{S})=\phi^{\ast}\mathrm{curv}(\Xi)-\mathrm{curv}(\Xi)=0$, and
as $H_{\mathrm{loc}}^{1}(\Gamma(E))=0$ there exists $\sigma_{\phi}\in
\Omega_{\mathrm{loc}}^{0}(\Gamma(E))$ such that $\phi^{\ast}\rho^{S}-\rho
^{S}=d\sigma_{\phi}$. By Proposition \ref{localityAlfa} if $\gamma$ is a curve
joining $s_{0}$ and $s$ we have $\alpha_{\phi}^{S}(s)=\alpha_{\phi}^{S}%
(s_{0})+\int_{\gamma}(\phi^{\ast}\rho^{S}-\rho^{S})=\alpha_{\phi}^{S}%
(s_{0})+\int_{\gamma}d\sigma_{\phi}=\alpha_{\phi}^{S}(s_{0})+\sigma_{\phi
}(s)-\sigma_{\phi}(s_{0})$. The result follows because the first and third
terms are local by Lemma \ref{Lemma}.
\end{proof}

We know that $\alpha^{S}$ satisfies the cocycle condition. And by Assumption
(A3) and Propositions \ref{localSection}\ and \ref{curvatura} c)\ the
cohomology class of $\alpha^{S}$ on $H^{1}(\mathcal{G},\Omega_{\mathrm{loc}%
}^{0}(\Gamma(E)\mathbb{)}/\mathbb{Z})$ does not depend on the $\Xi$-local
section chosen. We denote this class by $\{\alpha^{\mathcal{U}}\}\in
H^{1}(\mathcal{G},\Omega_{\mathrm{loc}}^{0}(\Gamma(E)\mathbb{)}/\mathbb{Z})$,
and we have the following

\begin{theorem}
\label{ThGlobalPhysical}If $\Xi$ is a $\mathcal{G}$-invariant local
connection, then the following conditions are equivalent

G$_{1}$) There exists a $\Xi$-local $\mathcal{G}$-equivariant section of
$\mathcal{U}\rightarrow N$.

G$_{2}$) $\{\alpha^{\mathcal{U}}\}=0$ on $H^{1}(\mathcal{G}$, $\Omega
_{\mathrm{loc}}^{0}(\Gamma(E)\mathbb{)}/\mathbb{Z})$.

G$_{3}$) There exists $\beta\in\Omega_{\mathcal{G},\mathrm{loc}}^{1}%
(\Gamma(E))$ such that $\mathrm{hol}_{\phi}^{\Xi}(\gamma)=\int_{\gamma}\beta$
for any $\phi\in\mathcal{G}$, and $\gamma\in\mathcal{C}^{\phi}$.
\end{theorem}

\begin{proof}
G$_{1}$)$\Rightarrow$G$_{2}$) If $S\colon\Gamma(E)\rightarrow\mathcal{U}$ is a
$\Xi$-local $\mathcal{G}$-equivariant section then $\alpha^{S}=0$ and hence
$[\alpha^{\mathcal{U}}]=0$ on $H^{1}(\mathcal{G}$, $\Omega_{\mathrm{loc}}%
^{0}(\Gamma(E)\mathbb{)}/\mathbb{Z})$.

G$_{2}$)$\Rightarrow$G$_{3}$) If $\{\alpha^{\mathcal{U}}\}=0$ on
$H_{\mathrm{loc}}^{1}(\mathcal{G},\Omega_{\mathrm{loc}}^{0}(M,\mathbb{R}%
)/\mathbb{\mathbb{Z})}$ we chose a $\Xi$-local section $S\colon N\rightarrow
\mathcal{U}$ and we have $\alpha^{S}=\phi_{N}^{\ast}\theta-\theta$ for
$\theta\in\Omega_{\mathrm{loc}}^{0}(M,\mathbb{R)}$. Hence $\mathrm{hol}_{\phi
}^{\Xi}(\gamma)=\int_{\gamma}\rho^{S}+\phi_{N}^{\ast}\theta-\theta
=\int_{\gamma}(\rho^{S}+d\theta)$. We define $\beta=\rho^{S}+d\theta$ and
using Proposition \ref{localityAlfa}\ we have $\phi^{\ast}\beta=\phi^{\ast
}\rho^{S}+\phi^{\ast}d\theta=\phi^{\ast}\rho^{S}+d\alpha^{S}+d\theta=\rho
^{S}+d\theta=\beta$, and hence $\beta$ is $\mathcal{G}$-invariant.

G$_{3}$)$\Rightarrow$G$_{1}$) If $\beta\in\Omega_{\mathcal{G},\mathrm{loc}%
}^{1}(\Gamma(E))$ satisfies $\mathrm{hol}_{\phi}^{\Xi}(\gamma)=\int_{\gamma
}\beta$ we have $\mathrm{curv}(\Xi)=d\beta\in\Omega_{\mathrm{loc}}^{2}%
(\Gamma(E))$, and by Proposition \ref{triviality}\ there exists a section $S$
such that $\rho^{S}=\beta\in\Omega_{\mathcal{G},\mathrm{loc}}^{1}(\Gamma(E))$.
By Proposition \ref{loghol}\ we have $\alpha^{S}=0$, and hence $S$ is
$\mathcal{G}$-equivariant. Finally, $S$ is $\Xi$-local as $\mathrm{curv}%
(\Xi)=d\beta\in\Omega_{\mathrm{loc}}^{2}(\Gamma(E))$\ and we have $\mu^{\Xi
}(X)=\rho^{S}(X_{N})+\mathfrak{a}^{S}(X)=\rho^{S}(X_{N})\in\Omega
_{\mathrm{loc}}^{0}(\Gamma(E))$, where we have used that $\mathfrak{a}^{S}=0$
as $\alpha^{S}=0$. We conclude that $\mathrm{curv}_{\mathcal{G}}(\Xi)\in
\Omega_{\mathcal{G},\mathrm{loc}}^{2}(\Gamma(E))$.
\end{proof}

By definition any $\mathcal{G}$-flat connection $\Xi$ is local, and any
section $S$ such that $\rho^{S}=0$ is $\Xi$-local. If we define
$K_{\mathrm{loc}}^{\mathcal{G}}(\Gamma(E))=k(H_{\mathcal{G},\mathrm{loc}}%
^{1}(\Gamma(E)))\subset\mathrm{Hom}(\mathcal{G}/\mathcal{G}_{0},\mathbb{R}%
/\mathbb{Z})$ then we have the following

\begin{proposition}
\label{flatLocal}If $\Xi$ is a $\mathcal{G}$-flat connection on $\mathcal{U}%
\rightarrow\Gamma(E)$, then there exists a $\Xi$-local and $\mathcal{G}%
$-equivariant section if and only if $\kappa^{\Xi}\in K_{\mathrm{loc}%
}^{\mathcal{G}}(\Gamma(E))$.
\end{proposition}

\section{Concluding remarks\label{CR}}

We have obtained necessary and sufficient conditions for physical anomaly
cancellation in terms of the existence of a local $\mathcal{G}$-invariant$\ 1$%
-form $\beta$\ satisfying certain conditions. The advantage of this approach
with respect to the BRST cohomology is that $\Omega_{\mathrm{loc}}^{1}%
(\Gamma(E))^{\mathcal{G}}$ can be completely determined in terms of the
variational bicomplex of $J^{\infty}E$ and its cohomology can be related to
Gelfand-Fuks cohomology of formal vector fields (see \cite{VB}). This
technique is applied in \cite{LocUni} to study physical gravitational
anomalies. It is proved that if the equivariant holonomy does not vanish, it
is possible to cancel the physical gravitational anomaly only in dimensions
$n=3\operatorname{mod}4$ and by means of a Chern-Simons counterterm.

The problem of global anomaly cancellation in dimension $3$ has been recently
reanalyzed by\ Witten in \cite{Witten2016} and \cite{WittenFermionic}. In
these papers it is conjectured that to cancel the anomaly it is not sufficient
with the existence of a invariant partition function. To have an anomaly free
theory consistent with the principles of unitarity, locality and cutting and
pasting, the phases of the partition functions for different $3$-manifolds
should also be fixed. We hope\ that our geometric characterization of anomaly
cancellation can be used to study this problem.

\end{document}